\documentclass[12pt]{article}
\usepackage{fullpage}
\usepackage{amsmath, amsfonts}
\newtheorem{theorem}{Theorem}[section]
\newtheorem{corollary}[theorem]{Corollary}

\newtheorem{definition}[theorem]{Definition}
\newtheorem{lemma}[theorem]{Lemma}
\newcommand{\qed}{\rule{6pt}{6pt}}
\newenvironment{proof}{\noindent{\bf Proof:}}{\qed\bigskip}
\title{The Two\,Bicliques Problem is in NP\,$\cap$\,coNP}
\author{M. A. Shalu \hspace{2em}  S. Vijayakumar\\
	Indian Institute of Information Technology\\ Design \& Manufacturing 
	(IIITD\&M) Kancheepuram \\ IIT Madras Campus\\ Chennai 600036, India\\
	$\{\mbox{shalu, vijay}\}$@$\mbox{iiitdm.ac.in}$}
\begin{document}
\maketitle
\begin{abstract}
We show that the problem of deciding whether the vertex set of  a graph can be covered with at most two bicliques is in NP$\cap$coNP. We thus almost determine the computational complexity of a problem whose status has remained open for quite some time. Our result implies that a polynomial time algorithm for the problem is more likely than it being  NP-complete unless P = NP.\bigskip \\
{\bf keywords:} Bicliques, Polynomial Time Algorithms, NP, coNP
\end{abstract}

\section{Introdution}
The problem of covering {\em the vertex set} of a graph with a minimum number of bicliques is one of the basic problems of graph theory with numerous applications of both theoretical and  practical importance~\cite{hms01,mrs04,mou04,nj05,njbj03,pee03}. Heydari, Morales, Shields Jr., and Sudborough show that the corresponding decision problem of determining whether a graph can be covered with at most $k$ bicliques is NP-complete~\cite{hmss07}. Indeed, Fleischner, Mujuni, Paulusma, and Szeider show that this decision problem remains NP-complete even when $k$ is a fixed integer greater than two and not part of the input~\cite{fmps09}.

Interestingly, the complexity of deciding whether the vertex set of a graph can be covered with at most two bicliques has remained a challenging open problem. In particular, any theoretical evidence in favor of the problem either having an efficient algorithm or being NP-complete has remained elusive; see, for instance, \cite{bbmmss08,dms11,figu11,fmps09,hmss07}. In fact, Figueiredo classifies this problem, among a few others, as one of the important problems even in the P versus NP arena~\cite{figu11}.

In this paper,  we establish that this problem is in NP$\cap$coNP. This effectively settles the problem in favor an efficient algorithm. For we learn from computational complexity theory that such a problem is least likely to be NP-complete. For otherwise, the polynomial hierarchy is known to collapse to the first level~\cite{gj79,pap94}. And problems that were seen to be in NP$\cap$coNP have invariably been found subsequently to be in P as well~\cite{pap94}. 

Despite the fact that the problem allows efficient algorithms for several special classes of graphs~\cite{bbmmss08,deffk08,dms11,fmps09}, our result still comes as a surprise for at least two reasons: (i) The closely related problem of deciding whether the vertex set of a connected graph can be covered with two $P_4$-free graphs is  shown to be NP-complete by  H\`{o}ang  and  Le  \cite{hl01}. (ii) Deciding whether a graph can be covered with two bicliques is essentially equivalent to deciding whether a connected graph has a disconnected vertex cut (see Lemma~\ref{bp2} or~\cite{fmps09}, for instance) but the closely related problem of deciding whether a connected graph has an independent vertex cut is known to be NP-complete~\cite{bdls00,chv84,lr03}.
[But a clique vertex cut is known to have a polynomial time algorithm~\cite{whi81}.] 

\paragraph{Note:} Covering the vertex set of a graph with a minimum number of bicliques turns to be equivalent to partitioning the vertex set of the underlying graph into a minimum number of parts so that the induced subgraph on each part is covered by exactly one biclique. Therefore, by {\em partitioning a grpah into a minimum number of bicliques}, we essentially mean covering the vertex set of the graph with a minimum number of bicliques. 

\paragraph{Notation:} We denote by {BP$k$} the set of all graphs $G$ such that $G$ can be {\em partitioned} into at most $k$ bicliques or, equivalently, such that the vertex set of $G$ can be covered with  at most $k$ bicliques. 

By $\overline{\mbox{BP$k$}}$, we denote the set of all graphs $G$ such that 
$G\notin\mbox{BP$k$}$. Equivalently, $\overline{\mbox{BP$k$}}$ is the set of all graphs $G$ such that every partition of $G$ into bicliques has more than $k$ parts. 

We use {BP} for denoting the set of all pairs $(G, k)$ such that the graph $G$ can be partitioned into at most $k$ bicliques.

By convention, we will use {BP$k$}, $\overline{\mbox{BP$k$}}$, and {BP} for denoting the membership problems associated with these sets.

\paragraph{Related Work:} Bein, Bein, Meng, Morales,  Shields Jr., and Sudborough show that  it is NP-hard to find a $c$-approximation algorithm for {BP} for any constant $c$, apart from presenting a polynomial time exact algorithm for {BP$k$} restricted to bipartite graphs and restricted to certain other families of graphs~\cite{bbmmss08}.

The result of Fleischner, Mujuni, Paulusma, and Szeider that {BP$k$} is NP-complete for each fixed $k \ge 3$ also rules out a fixed parameter tractable algorithm for {BP} unless P = NP~\cite{fmps09}. They moreover show that a certain natural bounded version of {BP} remains NP-complete and is W[2]-complete~\cite{df99}. In contrast, they  show the edge set version of biclique cover and biclique partition problems, which are known to be NP-complete \cite{jr93,mul96,orl77} to be fixed parameter tractable.   Their work includes a polynomial time algorithm for {BP2} restricted to a family of graphs that includes bipartite graphs. 

Recently, Dantas, Maffray, and Silva provide a list of several natural families of graphs such that there is a polynomial time algorithm for {BP2} when restricted to graphs in each of these families~\cite{dms11}. The list of families of graphs they consider includes $K_4$-free graphs, diamond-free graphs, planar graphs, bounded treewidth graphs, claw-free graphs, and $(C_5, P_5)$-free graphs.  

{\em Bicliques} are one of the most sought-after structures of graphs, mainly due to their importance in applications, and has  given rise to numerous computational problems involving bicliques from diverse branches of science; please consult the references. 

\section{Preliminaries}
In this paper, we consider finite undirected simple graphs. We begin by formally defining a biclique as well as a star of a graph.

\begin{definition}
\begin{enumerate}
\item A subgraph $H$ of a graph $G$ is said to be a {\em biclique} if $H$ is 
	isomorhic to either the complete graph $K_1$ or the complete bipartite 
	graph $K_{m,n}$ for some $m,n\ge 1$. 
\item A biclique $H$ of a graph $G$ is said to be a {\em star} if $H$ is 
	isomorphic to either the complete graph $K_1$ or the complete bipartite graph 
	$K_{1,n}$ for some $n\ge 1$. The {\em center} of a star $H$ is defined 
	naturally. 
\end{enumerate}
\end{definition}

We now review the standard graph theory terminology and notation that we use. 

\begin{definition}
\begin{enumerate}
\item  $\bar{G}$ denotes the complement of a graph $G$. 
\item The empty graph on $n$ vertices is denoted by $nK_1$: $ nK_1 
	= \bar{K}_n$.
\item For a graph $G = (V, E)$ and $v\in V$, $N_G(v)$ denotes 
	the set all vertices that are adjacent to $v$. [$N(v)$ does not include $v$.] 
	We define  $N_G[v] = N_G(v) \cup \{v\}$. We use $N(v)$ and $N[v]$ for 
	these sets when $G$ is understood. 
\item For a graph $G = (V, E)$ and a set $A\subseteq V$, $G[A]$ 
	denotes the induced subgraph of $G$ on the vertices of $A
	$.
\item For a graph $G$ and a vertex  $v$ of it, $G-v$ denotes the 
	induced subgraph on $V(G) \setminus \{v\}$.
\item For a graph $G = (V, E)$ and a set $A\subseteq V(G)$, 
	$G - A$ denotes the induced graph on $V(G) \setminus A$.
\item A vertex $v$ of a connected graph $G$ is said to be a {\em 
	cut vertex} if $G - v$ is disconnected.
\item A set $X$ of vertices of a connected graph $G$ is said to 
	be a {\em vertex cut} if $G - X$ is disconnected.
\end{enumerate}
\end{definition}

We record a simple characterization of BP2 that is in the folklore. We state and prove it for compleness. Naturally, it turns to be a characterization for $\overline{\mbox{BP2}}$ as well. We begin with the following.

\begin{lemma}
\label{lemma-bp1}
A graph $G\neq K_1$ is in {BP1} if and only if $\bar{G}$ is disconnected.
\end{lemma}
\begin{proof}
Let $G\in \mbox{BP1}$. Then it possible that $G = K_1$; otherwise let $[A,B]$  be a partion of $V(G)$ such that each vertex of $A$ is connected to every vertex of $B$. Then the complement graph $\bar{G}$ has no vertex of $A$ connected to any vertex of $B$. 

Conversely, if $G = K_1$ then it is a trivial biclique and belongs to BP1. Otherwise, assume that $\bar{G}$ is disconnected and set 
$A$ to the set of vertices of a connected component of $\bar{G}$ and $B$ to $V(G) \setminus A$. It follows that there is a biclique structure across $A$ and $B$ and so $G\in \mbox{BP1}$. 
\end{proof}

\begin{lemma}
\label{bp2n1}
A graph $G \neq 2K_1$ is in $\mbox{BP2}\setminus \mbox{BP1}$ if and only if $\bar{G}$ is connected but has either a cut vertex or a disconnected vertex cut.
\end{lemma}
\begin{proof}
Let $G$ be a graph such that $G = 2K_1$ or $\bar{G}$ is connected but has a cut vertex or a disconnected vertex cut. Since $G = 2K_1 \in \mbox{BP2}\setminus \mbox{BP1}$, we shall assume that $G \neq 2K_1$ and that $\bar{G}$ is connected. 
Then, clearly $G\notin\mbox{BP1}$ by Lemma~\ref{lemma-bp1}.

If $\bar{G}$ has a cut vertex,  say $v$, then $\overline{G - v} = \bar{G} - v$ is disconnected and therefore, by Lemma~\ref{lemma-bp1}, $G-v$ belongs to $\mbox{BP1}$. So, we conclude that
$G\in \mbox{BP2}\setminus \mbox{BP1}$.

If $\bar{G}$ has a disconnected vertex cut $C$, i.e., $C$ is a vertex cut of $\bar{G}$ such that both $\bar{G}[C]$  and $\bar{G}[V(G)\setminus C]$ are disconnected, then both $G[C]$ and  $G[V(G)\setminus C]$ are in {BP1} by Lemma~\ref{lemma-bp1}.  So, we again conclude that $G\in\mbox{BP2}\setminus \mbox{BP1}$.

Conversely, suppose that  $G\in \mbox{BP2}\setminus\mbox{BP1}$ and is not equal to $2K1$. Then $\bar{G}$ is necessarily connected; otherwise $G\in\mbox{BP1}$ by Lemma~\ref{lemma-bp1}. 

If $G$ has a two biclique partition with one of the parts as a single vertex, say $v$, then $G-v$ can be covered with one biclique which implies that $\bar{G} - v$ is disconnected, where we started with a $\bar{G}$ that is connected. Therefore $v$ must be a cut vertex of $\bar{G}$.

If $G$ allows a two biclique partition where neither of the bicliques is a single vertex, then $\bar{G}$ must be partitionable into two sets $A$ and $B$ such that both $A$ and $B$ have at least
two elements each and $\bar{G}[A]$ and $\bar{G}[B]$ are disconnected. But $\bar{G} = \bar{G}[A\cup B]$ is connected. Therefore, it must be that  $A$ (as well as $B$) is a disconnected vertex cut of $A$. 
\end{proof}

Combining the preceding lemmas, we have the following.

\begin{lemma}
\label{bp2}
A graph $G$ that is not equal to $K_1$ or $2K_1$ is in \mbox{BP2} if and only if one of the following is true: (a) $\bar{G}$ is disconnected; (b) $\bar{G}$ is connected but has a cut vertex; 
(c) $\bar{G}$ is connected but has a disconnected vertex cut.
\end{lemma}

Consequently, we have the following lemma for graphs not in {BP2}.

\begin{lemma}
\label{nbp2}
A graph $G$ on $n \ge 3$ vertices is in $\overline{\mbox{BP2}}$ if and only if $\bar{G}$ is connected, is free of cut vertices, and has all vertex cuts (if any) connected.
\end{lemma}

The corollary below follows trivially from the lemma. 

\begin{corollary}
\label{cor-main}
Let $G$ be a graph in $\overline{\mbox{BP2}}$. Then the following are true for the complement graph $\bar{G}$. 
\begin{enumerate}
\item The neighbours of any vertex of $\bar{G}$  induces a connected subgraph of $\bar{G}$ and this subgraph has at least two vertices.
\item From any vertex of $\bar{G}$, all other vertices are at most at a distance of two. 
\item Any nonadjacent pair of vertices of $\bar{G}$ have a common neighbour in $\bar{G}$.
\end{enumerate}
\end{corollary}

We close the section with a definition that encapsulates an important notion that is central to our discussion.

\begin{definition}
Let ${\bf F}$ be a family of graphs and let $G\in \mbox{\bf F}$. Let $\pi$ be a permutation of a set $A\subseteq V(G)$ with $|A| = k$. Then $\pi$ is said to be {\em safe for {\bf F}} if each of $G_0, G_1, G_2, \ldots, G_k\in \mbox{\bf F}$, where $G_i$ is the graph obtained from $G$ by deleting all the vertices in a prefix of length $i$ of $\pi$ for each $0\le i\le k$. 
\end{definition}

\section{Graphs of $\mbox{BP2}\setminus \mbox{BP1}$}
We show that from any graph $G$ in $\mbox{BP2}\setminus\mbox{BP1}$, by repeated deletion of zero or more vertices, we eventually and {\em inescapably} end up with a graph $G'$ in $\mbox{BP2}\setminus\mbox{BP1}$ that admits  a partition into a star and a biclique, without ever leaving $\mbox{BP2}\setminus\mbox{BP1}$ in the process. But we begin by proving the following Theorem. 

\begin{theorem}
\label{star-biclique}
Let $G$ be a graph in $\mbox{BP2}\setminus \mbox{BP1}$. Then we can decide whether $G$ allows a star-biclique partition in polynomial time. 
\end{theorem}
\begin{proof}
Let $G$ be a graph in $\mbox{BP2}\setminus \mbox{BP1}$. Then for each vertex $v$ of $G$, we simply check whether $G$ admits a partition into a star biclique {\em centered at $v$} and another biclique. We do this as follows by fixing $v$ for a particular vertex of $G$. 

If $G$ is disconnected, then there must be exactly two components. We simply check if at least one of the components is a star with $v$ as the center; this can be done in polynomial time.  So, we shall assume that $G$  is connected. 

If ${G-v}\in \mbox{BP1}$, then $v$ and $G-v$ provides a star-biclique partition of $G$. If $G - N[v] \in \mbox{BP1}$, then $G[N[v]]$ and $G-N[v]$ provides a star-biclique partition of $G$. 

If neither is the case, we decide in polynomial time  whether there is a proper subset $S\neq \emptyset$  of $N_G(v)$ such that deleting $\{v\}$ and $S$ from $G$ results in a graph in {BP1}. For if there is such an $S$, then $G[\{v\}\cup S]$ and $G -v - S$ provides a star-biclique partition.

Since neither $G-v$ nor $G - N_G[v]$ is in $\mbox{BP1}$, both $G -v$ and $G - N_G[v]$ contain at least two vertices and the complement graphs $\bar{G}-v$ and $\bar{G} - N_G[v]$ are connected. Let $A = N_G(v)$ and let $B = V(G) \setminus N_G[v]$. Clearly, $A\cup B = V(G)\setminus \{v\}$.

Consider the complement graph $\bar{G}-v$. Let $S$ be the set of all vertices $u$ in $A$ such that $u$ is adjacent to some vertex in $B$ in this complement graph. We note that this $S$ can be constructed in polynomial time. If $S = A$, [i.e., if each vertex of $A$ is adjacent to a vertex in the connected graph $\bar{G} - N_G[v]$], then deleting no subset of $A$ can disconnect $\bar{G}-v$; we shall therefore conclude that it is impossible to partition $G$ into a star centered at $v$ and a biclique. 

If $S\neq A$, then $S$ is a vertex cut for $\bar{G}-v$ and $\{v\}\cup S$ is a disconnected vertex cut for $\bar{G}$ with $v$ as a component (No vertex in $S\subseteq A = N_G(v)$ is adjacent to $v$ in $\bar{G}$.). In this case, we see that $G[\{v\}\cup S]$ and $G - v - S$ provide a star-biclique partition of $G$.  
\end{proof}

We have the following interesting result about graphs of $\mbox{BP2}\setminus \mbox{BP1}$ that do not admit a star-biclique partition.

\begin{lemma}
\label{bp2n1-deletion}
Let $G$ be a graph in $\mbox{BP2}\setminus \mbox{BP1}$ such that it does not admit any star-biclique partition. Then for any vertex $v$ of $G$, $G - v$ is also a graph in BP2$\setminus$BP1.
\end{lemma}
\begin{proof}
Suppose that $G$ does not allow any two biclique partition for which one of the bicliques is a star. 

Then each biclique in every two biclique partition of $G$ has on each side at least two vertices. 
So, deleting a vertex $v$ from $G$ does still retain a two biclique structure in $G - v$; and so $G - v \in \mbox{BP2}$.

Since assuming that $G - v \in \mbox{BP1}$ implies that $G$ admits a star-biclique partition, namely $v$ and $G-v$, we conclude that $G - v \in \mbox{BP2}\setminus \mbox{BP1}$.
\end{proof}

The following theorem is a corollary of the above lemma.

\begin{theorem}
\label{thm-bp2n1-proof}
For each graph $G$ in $\mbox{BP2}\setminus\mbox{BP1}$, there is an integer $l = l(G)\ge 0$ such that any permutation $\pi$ of any subset of $l$ vertices of $G$ is safe for $\mbox{BP2}\setminus\mbox{BP1}$. Moreover, none of the associated graphs $G_0, G_1, G_2, \ldots, G_{l-1}$ allows a star-biclique partition whereas the graph $G_l$ does.
\end{theorem}

\section{Graphs of $\overline{\mbox{BP2}}$}
The following theorem asserts that for any graph $G\in\overline{\mbox{BP2}}$, there is a careful order of deletion of vertices from $G$ so that each of the successively resulting subgraphs is in $\overline{\mbox{BP2}}$ and the last graph $H$ obtained is the smallest graph in $\overline{\mbox{BP2}}$, namely $3K_1 = \bar{K}_3$. 

\begin{theorem}
\label{thm-nbp2-proof}
Let $G$ be a graph in $\overline{\mbox{BP2}}$ on $n$ vertices. Then $G$ has a
permutation $\pi$ of $n-3$ vertices that is {\em safe} for $\overline{\mbox{BP2}}$. 
\end{theorem}
\begin{proof}
Let $G=(V,E)$ be a graph in $\overline{\mbox{BP2}}$ on $n$ vertices. We will construct a permutation $\pi = \langle v_1, v_2, \ldots, v_{n-3} \rangle$ of $n-3$ vertices of $G$ that is {\em safe} for $\overline{\mbox{BP2}}$: deleting vertices in any prefix of $\pi$ from $G$ leaves behind a graph in $\overline{\mbox{BP2}}$.

Let $A$ be a subset of $V$ of largest cardinality such that the induced subgraph $G[A]\in \mbox{BP2}$. In fact, the maximality of $A$ implies that $G[A] \in \mbox{BP2} \setminus \mbox{BP1}$. Let $v\in V\setminus A$. Then $G[A\cup \{v\}]\in \overline{\mbox{BP2}}$. Clearly, deleting vertices in $V\setminus (A\cup \{v\})$ from $G$, in any order, can never result in a graph in BP2. We set $\pi'$ equal to some ordering of vertices in  $V\setminus (A \cup \{v\})$.

For every partition  $[A_1, A_2]$ of $A$ such that both $G[A_1]$ and $G[A_2]$
are in BP1, we have at least one vertex in $A_1$ that is not adjacent to 
$v$ and at least one vertex in $A_2$ that is not adjacent to $v$.
In fact, we have that $G[A_1\cup\{v\}]\notin\mbox{BP1}$ and that $G[A_2\cup\{v\}]\notin\mbox{BP1}$. For otherwise we will have that $G[A\cup \{v\}]\in {\mbox{BP2}}$.

Let $B$ be a subset of $A$ of largest cardinality such that both $G[B]$ and $G[C]$, where $C = A \setminus B$, are in BP1. Then it follows, from the maximality of $B$ that for each $c\in C$, there is at least one vertex $b\in B$ such that $c$ is not adjacent to $b$. From what we noted in the preceding paragraph it also follows that $v$ is not adjacent to some vertex in $B$ and to some vertex in $C$. 

We now delete all vertices in $C$ that are adjacent to $v$ in some order. It is clear that the sequence of successive graphs that are resulting are all in $\overline{\mbox{BP2}}$. We continue deleting the other vertices of $C$ except for one, say $u$, and note again that the successively resulting graphs are all in $\overline{\mbox{BP2}}$. Let $p''$ denote the sequence of vertices deleted in the order of deletion. Let $H = G[B\cup \{u\} \cup \{v\}]$ denote the final graph obtained. 

We note that vertices $v$ and $u$ are not adjacent in $H = G[B\cup \{u\} \cup \{v\}]$. Both $v$ and $u$ have nonadjacent vertices in $B$. Delete in some order
all the vertices in $B$ adjacent to $v$ or $u$ or both from $H$. When this is done, vertices $v$ and $u$ become isolated. We now continue deleting the other vertices of $B$ except for one, say $w$, in some order. Let $\pi'''$ be the sequence of vertices deleted. It is clear again that all the graphs obtained after each additional deletion are all in $\overline{\mbox{BP2}}$.

We now set $\pi = \pi'\cdot \pi''\cdot\pi'''$ and see that $\pi$ is a sequence of $n-3$ vertices of $G\in \overline{\mbox{BP2}}$ on $n$ vertices and that $\pi$ is safe for $\overline{\mbox{BP2}}$. 
\end{proof}

\section{Proving that  $\mbox{BP2}\in\mbox{coNP}$}
We establish that BP2 is in coNP by showing that $\overline{\mbox{BP2}}$ is in NP. We provide a polynomial time verifier that takes in as input a graph $G$ and a sequence $\pi$ of vertices of $G$. The verifier accepts the pair if and only if $G\in \overline{\mbox{BP2}}$ and $\pi$ is safe for $\overline{\mbox{BP2}}$ and is of length $n-3$, where $n = |V(G)|$. We know, from  Theorem~\ref{thm-nbp2-proof}, that such a proof exists for all graphs in $\overline{\mbox{BP2}}$.

\begin{theorem}
\label{verifier}
There is a polynomial time algorithm that inputs a pair $(G, \pi)$ of a graph $G$ and a sequence $\pi$ of vertices of $G$ and outputs {\sc accept} if and only if $G\in \overline{\mbox{BP2}}$ and $\pi$ is a longest permutation of vertices 
of $G$ that is safe for $\overline{\mbox{BP2}}$; it otherwise outputs {\sc reject}.   
\end{theorem}
\begin{proof}
Consider the algorithm in Figure~1.  We argue that this algorithm provides a valid polynomial time verifier  for $\overline{\mbox{BP2}}$. It is clear, from Theorem~\ref{star-biclique}, that  the algorithm can run in polynomial time. We will just prove its correctness.

\begin{figure}[h]
\label{nbp2-verifier}
\hrule
Input: $(G, \pi)$

Output: {\sc accept / reject}
\hrule
\begin{enumerate}
\item[0.] If $G\in \mbox{BP1}$ or $\pi$ is {\em not} a permutation on $n-3$ 	vertices of the $n$ vertex graph $G$, return {\sc reject}.
 	Else {\bf repeat} Steps 1 to 3 below:
\item[1.] If $G$ admits a star-biclique partition, return {\sc 	reject}.
\item[2.] If $G = 3K_1$, return {\sc accept}.
\item[3.] Remove the first vertex, $v$, from $\pi$ and set $G = G 
	- v$.  
\end{enumerate}
\hrule
\caption{A Polynomial Time Verifier for $\overline{\mbox{BP2}}$}
\end{figure}

Suppose that $(G, \pi)$ is input to the algorithm. 

If either $G\in{\mbox{BP1}}$ or $\pi$ is not obviously a longest safe sequence, the pair $(G, \pi)$ is rightly rejected in Step~0. 

If $G\in{\mbox{BP2}}\setminus \mbox{BP1}$, then any repeated removal of zero or more vertices from $G$ eventually necessarily results in a graph $H$ that allows a star-biclique partition (Theorem~\ref{thm-bp2n1-proof}) before giving rise to any graph that is probably not in {BP2}. Step~1 therefore ensures that no graph $G\in{\mbox{BP2}}\setminus \mbox{BP1}$ ever leads to the acceptance of the pair $(G, \pi)$ with any false safe sequence  $\pi$ by detecting as and when a star-biclique structure arises from such a $G$; we know from Theorem~\ref{star-biclique} that this deduction can be carried out in polynomial time. 

If $G\in\overline{BP2}$ but $\pi$ is not safe for $\overline{BP2}$, then $\pi$ has a prefix whose removal from $G$ results in a graph $H$ in BP2. If $H$ does not admit a star-biclique partition, then continuing the removals further must (as argued in the preceding paragraph) eventually result in a graph that admits such a partition before possibly resulting in a graph that is not in BP2. Step~2 therefore also ensures that no wrong safe sequence $\pi$ even with a $G\in\overline{BP2}$ leads to the acceptance of $(G, \pi)$.

If $G$ is a graph in  $\overline{\mbox{BP2}}$ on $n$ vertices and $\pi$ is a permutation of $n-3$ vertices of $\pi$ that is safe for $\overline{\mbox{BP2}}$ (such a sequence exists from Theorem~\ref{thm-nbp2-proof}), then $\pi$ is necessarily a longest sequence that is safe for $\overline{\mbox{BP2}}$ and each subgraph of $G$ obtained by deleting a prefix of $\pi$ is in  $\overline{\mbox{BP2}}$ and so none of them can clearly allow a star-biclique partition. Moreover, deleting all the vertices from such a $\pi$ must necessarily result in $3K_1$; for this is the only graph on three vertices that is in $\overline{\mbox{BP2}}$. Therefore, such an input pair $(G, \pi)$ is eventually rightly accepted, as can be easily verified, in Step~2 of the algorithm. 

Steps~3 simply deletes the next vertex in $\pi$ from $G$. The sequence $\pi$ cannot be empty when the control enters Step~3 because it must have at least four vertices. For, if it has only three vertices, it  must have either allowed a star-biclique partition already or been equal to $3K_1$ already; and the algorithm would have already stopped with an {\sc accept} or a {\sc reject}.  
\end{proof}


\section*{Conclusion} It remains an interesting open problem to see if the two biclique partition problem has a polynomial time algorithm. A negative answer to it, in particular, will resolve the famous P versus NP problem.

\end{document}